\documentclass[onecolumn]{revtex4-2}
\usepackage[
	bookmarks  = false,
	colorlinks = true,
	linkcolor  = blue,
	urlcolor   = blue,
	citecolor  = blue]{hyperref}
\usepackage{amsmath,amssymb,amsthm,graphicx}
\usepackage{braket,bm,bbm}

\newtheorem{theorem}{Theorem}

\newtheorem{corollary}{Corollary}
\newtheorem{definition}{Definition}
\newtheorem{proposition}{Proposition}
\newtheorem{example}{Example}
\DeclareMathOperator{\Tr}{Tr}
\newcommand{\ketbra}[2]{\ket{#1}\!\bra{#2}}

\begin{document}
\title{Multi-player quantum data hiding by nonlocal quantum state ensembles}
\author{Donghoon Ha}
\affiliation{Department of Applied Mathematics and Institute of Natural Sciences, Kyung Hee University, Yongin 17104, Republic of Korea}
\author{Jeong San Kim}
\email{freddie1@khu.ac.kr}
\affiliation{Department of Applied Mathematics and Institute of Natural Sciences, Kyung Hee University, Yongin 17104, Republic of Korea}
\begin{abstract}
We provide multi-player quantum data hiding based on nonlocal quantum state ensembles arising from multi-party quantum state discrimination.
Using bounds on local minimum-error discrimination of multi-party quantum states, we construct a multi-player quantum data-hiding scheme.
Our data-hiding scheme can be used to hide multiple bits, asymptotically, unless all the players collaborate.
We also illustrate our results by examples of nonlocal quantum state ensembles.
\end{abstract}
\maketitle
\section{Introduction}
Quantum state discrimination is one of the most fundamental tasks in quantum information processing\cite{chef2000,berg2007,barn20091,bae2015}.
In general, orthogonal quantum states can be perfectly discriminated by using an appropriate measurement.
On the other hand, nonorthogonal quantum states cannot be perfectly discriminated using any measurement.
For these reasons, much attention has been shown for optimal strategies to discriminate nonorthogonal quantum states\cite{hels1969,ivan1987,diek1988,pere1988,crok2006}.

In multi-party quantum systems, there are orthogonal quantum states that cannot be perfectly discriminated only by \emph{local operations and classical communication}(LOCC)\cite{benn19991,ghos2001,walg2002,chit20142}.
As orthogonal quantum states can be perfectly discriminated by global measurements, non-perfect discrimination of multi-party orthogonal states only by LOCC measurements implies quantum nonlocality in state discrimination.
This nonlocal phenomenon also occurs in discriminating nonorthogonal quantum states which cannot be optimally discriminating using only LOCC\cite{pere1991,duan2007,chit2013}.
To characterize the nonlocality in quantum state discrimination, there has been a lot of research effort focused on finding nonlocal quantum state ensembles as well as possible applications in quantum information processing\cite{fan2004,duan2009,chid2013,terh20011,divi2002}. 

Quantum data hiding is one representative application of quantum nonlocality of multi-party quantum state discrimination\cite{terh20011,divi2002}.
In quantum data hiding, information is hidden 
from multiple players by using multi-party quantum states.
The hidden data can be perfectly recovered by global measurements that can be achieved only if all the players collaborate, whereas LOCC measurements cannot recover any meaningful information about the data.
So far, there have been several research results focused on nonlocal quantum state ensembles that can be used for quantum data hiding\cite{egge2002,mori2013,pian2014,peng2021,lami2021}.

Nevertheless, due to the lack of good mathematical structure for LOCC, it is still a hard task to clearly determine what type of quantum state ensemble can be used for quantum data hiding.
Recently, a relation between nonlocal quantum state ensemble and quantum data hiding processing was established based on bounds on optimal local discrimination of bipartite quantum states. 
Moreover, a sufficient condition was also provided for a bipartite quantum state ensemble to be used to construct a quantum data-hiding scheme\cite{ha20221,ha2024}.

Here, we establish multi-player data hiding based on nonlocal quantum state ensembles arising from multi-party quantum state discrimination.
Using bounds on optimal local discrimination of multi-party quantum states, we construct a multi-player quantum data-hiding scheme.
In our multi-player quantum data-hiding scheme, classical data of arbitrary size can be perfectly hidden, asymptotically, unless all the players collaborate.
Our results are illustrated by examples of multi-party quantum state ensembles.

Our data-hiding scheme utilizes orthogonal quantum state ensembles, ensuring perfect data recovery through the orthogonality of these states. 
To asymptotically reduce the information accessible about the hidden data only by LOCC measurements, we independently and repeatedly prepare a quantum state from a multi-party quantum state ensemble. 
As the number of repetitions increases, the accessible information decreases exponentially until it essentially reaches the level of random guessing. 
Consequently, our data-hiding scheme is capable of effectively concealing classical data with only a relatively small number of repetitions.
While such a scheme is well-established for two-party scenarios, its extension to cases involving more than two participants remains unexplored\cite{terh20011,divi2002,ha2024}.

This paper is organized as follows.
In Sect.~\ref{sec:pre}, we first recall the definitions and some properties about measurements in multi-party quantum systems.
We further review bounds on local minimum-error discrimination of multi-party quantum states and provide some properties useful in multi-player quantum data hiding. 
In Sect.~\ref{sec:qdh}, we establish a sufficient condition for a multi-party quantum state ensemble to be used to construct a multi-player quantum data-hiding scheme.
In Sect.~\ref{sec:ex}, we illustrate our results by examples in multi-party quantum systems.
In Sect.~\ref{sec:dsc}, we summarize our results with possible future works.

\section{Preliminaries}\label{sec:pre}
For a \emph{multi-party} quantum system consisting of $m\geqslant2$ subsystems $A_{1},\ldots,A_{m}$, a quantum state is expressed by a density operator $\rho$, that is, a positive-semidefinite operator $\rho\succeq0$ with $\Tr\rho=1$, acting on a multi-party Hilbert space 
$\mathcal{H}=\bigotimes_{k=1}^{m}\mathbb{C}^{d_{k}}$ with positive integers $d_{k}$ for $k=1,\ldots,m$.
A measurement is described by a positive operator-valued measure $\{M_{i}\}_{i}$, that is, a set of positive-semidefinite operators $M_{i}\succeq0$ on $\mathcal{H}$ satisfying $\sum_{i}M_{i}=\mathbbm{1}$, where $\mathbbm{1}$ is the identity operator on $\mathcal{H}$. 
When a measurement $\{M_{i}\}_{i}$ is performed to a quantum  state $\rho$, the probability of obtaining the measurement outcome corresponding to $M_{j}$ is $\Tr(\rho M_{j})$.
\subsection{Measurements in multi-party systems}\label{ssec:pame}
For an integer $k$ with $1\leqslant k\leqslant m$, we denote the set of all \emph{$k$-partitions} of $\{A_{1},\ldots,A_{m}\}$ as
\begin{eqnarray}\label{eq:tpx2}
\mathbb{X}_{k}=\{\{\mathbf{A}_{1},\ldots,\mathbf{A}_{k}\}&|&
\mathbf{A}_{1},\ldots,\mathbf{A}_{k}:
\mbox{mutually disjoint nonempty subsets of}~\{A_{1},\ldots,A_{m}\},\nonumber\\
&&\mathbf{A}_{1}\cup\cdots\cup\mathbf{A}_{k}=\{A_{1},\ldots,A_{m}\}\}.
\end{eqnarray}
We also denote by $\mathsf{G}$ the \emph{trivial} partition of $\{A_{1},\ldots,A_{m}\}$, that is,
\begin{equation}\label{eq:pag}
\mathsf{G}=\{\{A_{1},\ldots,A_{m}\}\},
\end{equation}
which is the unique element of $\mathbb{X}_{1}$.
The set of all partitions of $\{A_{1},\ldots,A_{m}\}$ is denoted by $\mathbb{X}$, that is,
\begin{equation}\label{eq:ntpx}
\mathbb{X}=\bigcup_{k=1}^{m}\mathbb{X}_{k},
\end{equation}
and the set of all \emph{nontrivial} partitions of $\{A_{1},\ldots,A_{m}\}$ as
\begin{equation}\label{eq:xc1}
\mathbb{X}_{1}^{\mathsf{c}}=\mathbb{X}\setminus\mathbb{X}_{1}=\bigcup_{k=2}^{m}\mathbb{X}_{k}.
\end{equation}

\begin{definition}\label{def:locc}
For a partition $\mathsf{X}\in\mathbb{X}$, we say that a measurement is a \emph{$\mathsf{X}$-LOCC measurement} if it can be realized by LOCC among the parties(elements) of $\mathsf{X}$, that is, local measurements on each party of $\mathsf{X}$ and classical communication among the parties of $\mathsf{X}$. 
We also say that a measurement is a \emph{non-global measurement} if it is a $\mathsf{X}$-LOCC measurement for some $\mathsf{X}\in\mathbb{X}_{1}^{\mathsf{c}}$.
\end{definition}
\noindent For a partition $\mathsf{X}\in\mathbb{X}$, we define $\mathbb{M}_{\mathsf{X}}$ to be the set of all $\mathsf{X}$-LOCC measurements.
In particular, $\mathbb{M}_{\mathsf{G}}$ is the set of all measurements that can be performed on the entire system. Note that a measurement is non-global if and only if it is in $\mathbb{M}_{\mathsf{X}}$ for some $\mathsf{X}\in\mathbb{X}_{1}^{\mathsf{c}}$.

For partitions $\mathsf{X},\mathsf{X}'\in\mathbb{X}$, $\mathsf{X}$ is called \emph{coarser} than $\mathsf{X}'$ if every element of $\mathsf{X}'$ is a subset of some element of $\mathsf{X}$.
Obviously, $\mathsf{G}$ is coarser than any $\mathsf{X}\in\mathbb{X}$.
For any $\mathsf{X}\in\mathbb{X}_{1}^{\mathsf{c}}$,
we also note that there exists a bipartition $\mathsf{X}'\in\mathbb{X}_{2}$ coarser than $\mathsf{X}$; 
two elements of $\mathsf{X}'$ can be constructed by fixing one element of $\mathsf{X}$ and the union of the rests.
If $\mathsf{X}$ is coarser than $\mathsf{X}'$, then 
Definition~\ref{def:locc} implies
$\mathbb{M}_{\mathsf{X}'}\subseteq\mathbb{M}_{\mathsf{X}}$.
Therefore, we have
\begin{equation}\label{eq:xcr}
\bigcup_{\mathsf{X}\in\mathbb{X}_{1}^{\mathsf{c}}}\mathbb{M}_{\mathsf{X}}=\bigcup_{\mathsf{X}\in\mathbb{X}_{2}}\mathbb{M}_{\mathsf{X}}
\subseteq
\mathbb{M}_{\mathsf{G}},
\end{equation}
where the equality is due to $\mathbb{X}_{1}^{\mathsf{c}}\supseteq\mathbb{X}_{2}$.

\subsection{Local minimum-error discrimination}\label{ssec:blmd}
For a multi-party quantum state ensemble,
\begin{equation}\label{eq:ens}
\mathcal{E}=\{\eta_{i},\rho_{i}\}_{i=0}^{n-1},
\end{equation}
let us consider the situation that the state $\rho_{i}$ is prepared with the probability $\eta_{i}$ for $i=0,\ldots,n-1$. 
In discriminating the states from the ensemble $\mathcal{E}$ using a measurement $\mathcal{M}=\{M_{i}\}_{i=0}^{n-1}$, 
the prepared state is guessed to be $\rho_{i}$ for the measurement outcome obtained from $M_{i}$.
In this case, the average probability of correctly guessing the prepared state is
\begin{equation}\label{eq:pcg}
\sum_{i=0}^{n-1}\eta_{i}\Tr(\rho_{i}M_{i}).
\end{equation}
For a partition $\mathsf{X}\in\mathbb{X}$, we denote by $p_{\mathsf{X}}(\mathcal{E})$ the maximum of success probability that can be obtained by using $\mathsf{X}$-LOCC measurements, that is,
\begin{equation}\label{eq:ple}
p_{\mathsf{X}}(\mathcal{E})=\max_{\mathcal{M}\in\mathbb{M}_{\mathsf{X}}}\sum_{i=0}^{n-1}\eta_{i}\Tr(\rho_{i}M_{i}).
\end{equation}
The \emph{minimum-error discrimination} of $\mathcal{E}$ is to achieve the maximum success probability $p_{\mathsf{X}}(\mathcal{E})$ for $\mathsf{X}=\mathsf{G}$, that is,
\begin{equation}\label{eq:pgdef}
p_{\mathsf{G}}(\mathcal{E})=\max_{\mathcal{M}\in\mathbb{M}_{\mathsf{G}}}\sum_{i=0}^{n-1}\eta_{i}\Tr(\rho_{i}M_{i}),
\end{equation}
where the maximum is taken over all possible measurements\cite{hels1969}.

From Eq.~\eqref{eq:xcr}, we have
\begin{equation}\label{eq:pxpx}
\max_{\mathsf{X}\in\mathbb{X}_{1}^{\mathsf{c}}}p_{\mathsf{X}}(\mathcal{E})
=\max_{\mathsf{X}\in\mathbb{X}_{2}}p_{\mathsf{X}}(\mathcal{E})
\leqslant p_{\mathsf{G}}(\mathcal{E}).
\end{equation}
We also note that deciding the prepared quantum state as the state with the largest probability is obviously a $\mathsf{X}$-LOCC measurement for any partition $\mathsf{X}\in\mathbb{X}$; thus, we have
\begin{equation}\label{eq:inq}
\max\{\eta_{0},\ldots,\eta_{n-1}\}\leqslant p_{\mathsf{X}}(\mathcal{E})
\end{equation}
for all partitions $\mathsf{X}\in\mathbb{X}$.

Given a multi-party quantum state ensemble $\mathcal{E}$ in Eq.~\eqref{eq:ens} and a bipartition $\mathsf{X}\in\mathbb{X}_{2}$, we provide an upper bound of $p_{\mathsf{X}}(\mathcal{E})$ in Eq.~\eqref{eq:ple} based on \emph{partial transposition} as well as some related properties useful in quantum data hiding.
For a multi-party state $\rho$ and a bipartition $\mathsf{X}\in\mathbb{X}_{2}$, we denote $\Gamma_{\mathsf{X}}(\rho)$ as the partial transposition of $\rho$ with respect to the bipartition $\mathsf{X}$\cite{pere1996,pptp}. 

Now, let us consider 
\begin{equation}\label{eq:qge}
q_{\mathsf{X}}(\mathcal{E})=\max_{\mathcal{M}\in\mathbb{M}_{\mathsf{G}}}\sum_{i=0}^{n-1}\eta_{i}\Tr[\Gamma_{\mathsf{X}}(\rho_{i})M_{i}],
\end{equation}
where the maximum is taken over all possible measurements.
This quantity is known as an upper bound of $p_{\mathsf{X}}(\mathcal{E})$\cite{ha20221}, that is,
\begin{equation}\label{eq:upb}
p_{\mathsf{X}}(\mathcal{E})\leqslant q_{\mathsf{X}}(\mathcal{E}).
\end{equation}
The following proposition provides a necessary and sufficient condition of a measurement realizing $q_{\mathsf{X}}(\mathcal{E})$\cite{ha20221}.

\begin{proposition}\label{prop:qgm}
For a multi-party quantum state ensemble $\mathcal{E}=\{\eta_{i},\rho_{i}\}_{i=0}^{n-1}$, a measurement $\mathcal{M}=\{M_{i}\}_{i=0}^{n-1}$ and a bipartition $\mathsf{X}\in\mathbb{X}_{2}$,
$\mathcal{M}$ realizes $q_{\mathsf{X}}(\mathcal{E})$ if and only if 
\begin{equation}\label{eq:nsqg}
\sum_{j=0}^{n-1}\eta_{j}\Gamma_{\mathsf{X}}(\rho_{j})M_{j}-\eta_{i}\Gamma_{\mathsf{X}}(\rho_{i})\succeq 0
\end{equation}
for each $i=0,\ldots,n-1$.
\end{proposition}

For an ensemble $\mathcal{E}=\{\eta_{i},\rho_{i}\}_{i=0}^{n-1}$ and a bipartition $\mathsf{X}\in\mathbb{X}_{2}$, the following theorem provides a necessary and sufficient condition for both $p_{\mathsf{X}}(\mathcal{E})$ and $q_{\mathsf{X}}(\mathcal{E})$ to be the probability $\eta_{0}$.

\begin{theorem}\label{thm:qxe0}
For a multi-party quantum state ensemble $\mathcal{E}=\{\eta_{i},\rho_{i}\}_{i=0}^{n-1}$ and a bipartition $\mathsf{X}\in\mathbb{X}_{2}$, we have 
\begin{equation}\label{eq:qxe0}
p_{\mathsf{X}}(\mathcal{E})=q_{\mathsf{X}}(\mathcal{E})=\eta_{0}
\end{equation}
if and only if
\begin{equation}\label{eq:trh}
\eta_{0}\Gamma_{\mathsf{X}}(\rho_{0})-\eta_{i}\Gamma_{\mathsf{X}}(\rho_{i})\succeq0
\end{equation}
for each $i=1,\ldots,n-1$.
\end{theorem}
\begin{proof}
Let $\mathcal{M}$ be the measurement $\{M_{i}\}_{i=0}^{n-1}$ satisfying
\begin{equation}\label{eq:mop}
M_{0}=\mathbbm{1},~M_{1}=\cdots=M_{n-1}=\mathbb{O}
\end{equation}
where $\mathbb{O}$ is the zero operator on $\mathcal{H}$.
We first suppose Eq.~\eqref{eq:qxe0}. Since $\mathcal{M}$ is obviously a measurement providing $q_{\mathsf{X}}(\mathcal{E})$, it follows from Proposition~\ref{prop:qgm} that Condition~\eqref{eq:trh} holds.

Conversely, let us assume Condition~\eqref{eq:trh}. This assumption implies that the measurement $\mathcal{M}$ satisfies Condition~\eqref{eq:nsqg}. Therefore, we have
\begin{equation}\label{eq:cope0}
q_{\mathsf{X}}(\mathcal{E})=\sum_{i=0}^{n-1}\eta_{i}\Tr[\Gamma_{\mathsf{X}}(\rho_{i})M_{i}]=\eta_{0},
\end{equation}
where the first equality is by Proposition~\ref{prop:qgm}.
Moreover, the first equality in Eq.~\eqref{eq:qxe0} holds because
\begin{equation}
\eta_{0}\leqslant p_{\mathsf{X}}(\mathcal{E})\leqslant q_{\mathsf{X}}(\mathcal{E})
=\eta_{0},
\end{equation}
where the first inequality is due to Inequality~\eqref{eq:inq}, the second inequality is by Inequality~\eqref{eq:upb}, and the last inequality follows from Eq.~\eqref{eq:cope0}.
\end{proof}
\noindent
As we can check in the proof of Theorem~\ref{thm:qxe0}, the choice of $\eta_{0}$ can be arbitrary. That is, any of $\{\eta_{i}\}_{i=0}^{n-1}$ can be used as long as Eq.~\eqref{eq:trh} holds in terms of $\eta_{i}$.

\section{Multi-player quantum data hiding}\label{sec:qdh}
In this section, we provide a sufficient condition for a multi-party quantum state ensemble to be used to construct a multi-player quantum data-hiding scheme.
To analyze how reliably the data are hidden, we first introduce the notion of \emph{multi-fold ensemble}.

For positive integers $n$ and $L$, let 
$\mathbb{Z}_{n}$ be the set of all integers from $0$ to $n-1$ and $\mathbb{Z}_{n}^{L}$ the Cartesian product of $L$ copies of $\mathbb{Z}_{n}$. 
For the multi-party quantum state ensemble $\mathcal{E}=\{\eta_{i},\rho_{i}\}_{i=0}^{n-1}$ in Eq.~\eqref{eq:ens} and each vector
\begin{equation}
\vec{c}=(c_{1},\ldots,c_{L})\in\mathbb{Z}_{n}^{L},
\end{equation}
we define
\begin{equation}\label{eq:erlf}
\eta_{\vec{c}}=\prod_{l=1}^{L}\eta_{c_{l}},~
\rho_{\vec{c}}=\bigotimes_{l=1}^{L}\rho_{c_{l}}
\end{equation}
where $\rho_{c_{l}}$ is the state in the ensemble $\mathcal{E}$ whose index is the $l$th coordinate of the vector $\vec{c}$, and $\eta_{c_{l}}$ is the corresponding probability for $l=1,\ldots,L$.
We further use $\mathcal{E}^{\otimes L}$ to denote the \emph{$L$-fold ensemble} of $\mathcal{E}$, that is,
\begin{equation}\label{eq:lfe}
\mathcal{E}^{\otimes L}=\{\eta_{\vec{c}},\rho_{\vec{c}}\}_{\vec{c}\in\mathbb{Z}_{n}^{L}}.
\end{equation}
For each $\vec{c}\in\mathbb{Z}_{n}^{L}$, we define $\omega_{n}(\vec{c})$ 
as the modulo-$n$ summation of all entries in $\vec{c}=(c_{1},\ldots,c_{L})$, that is,
\begin{equation}\label{eq:mns}
\omega_{n}(\vec{c})=\sum_{l=1}^{L}c_{l}~(\mathrm{mod}~n).
\end{equation}
Clearly, we have $\omega_{n}(\vec{c})\in\mathbb{Z}_{n}$ for any $\vec{c}\in\mathbb{Z}_{n}^{L}$.

For the state $\rho_{\vec{c}}$ prepared from the ensemble $\mathcal{E}^{\otimes L}$, let us consider the situation of guessing $\omega_{n}(\vec{c})$ of the prepared state $\rho_{\vec{c}}$ from $\mathcal{E}^{\otimes L}$.
This situation is equivalent to discriminating the states $\rho_{0}^{(L)},\ldots,\rho_{n-1}^{(L)}$ prepared with the probabilities $\eta_{0}^{(L)},\ldots,\eta_{n-1}^{(L)}$, respectively, where
\begin{eqnarray}
\eta_{i}^{(L)}=
\sum_{\substack{\vec{c}\in\mathbb{Z}_{n}^{L}\\ \omega_{n}(\vec{c})=i}}
\eta_{\vec{c}},~
\rho_{i}^{(L)}=\frac{1}{\eta_{i}^{(L)}}
\sum_{\substack{\vec{c}\in\mathbb{Z}_{n}^{L}\\ \omega_{n}(\vec{c})=i}}
\eta_{\vec{c}}
\rho_{\vec{c}}.\label{eq:rhil}
\end{eqnarray}
In other words, quantum state discrimination of the ensemble 
\begin{equation}\label{eq:eld}
\mathcal{E}^{(L)}=\{\eta_{i}^{(L)},\rho_{i}^{(L)}\}_{i=0}^{n-1}.
\end{equation}

In the data-hiding scheme described later, guessing the hidden $n$-ary classical data $x$ by $m$ players is equivalent to discriminating among the $n$ quantum states prepared from the ensemble $\mathcal{E}^{(L)}$ in Eq.~\eqref{eq:eld}. 
Moreover, the number of folds, $L$, in Eq.~\eqref{eq:eld} acts as the parameter that asymptotically conceals the data $x$. 
As $L$ increases, the information accessible about $x$ without full cooperative action among $m$ players decreases exponentially, eventually approaching the level of random guessing. 
In other words, the maximum average probability of discriminating the states from $\mathcal{E}^{(L)}$ using only non-global measurements converges exponentially to $1/n$ as $L$ increases. 
In the subsequent propositions and corollaries, we provide sufficient conditions that guarantee this convergence.

For a given multi-party quantum state ensemble $\mathcal{E}=\{\eta_{i},\rho_{i}\}_{i=0}^{n-1}$, a bipartition $\mathsf{X}\in\mathbb{X}_{2}$ and a positive integer $L$, the following proposition gives an upper bound of $p_{\mathsf{X}}(\mathcal{E}^{(L)})$\cite{ha2024}.

\begin{proposition}\label{prop:lqge}
For a multi-party quantum state ensemble $\mathcal{E}=\{\eta_{i},\rho_{i}\}_{i=0}^{n-1}$, a bipartition $\mathsf{X}\in\mathbb{X}_{2}$ and a positive integer $L$, we have
\begin{equation}\label{eq:qgel}
q_{\mathsf{X}}(\mathcal{E}^{(L)})\leqslant\frac{1}{n}+\frac{n-1}{n}(n\cdot q_{\mathsf{X}}(\mathcal{E})-1)^{L}.
\end{equation}
Moreover, if $q_{\mathsf{X}}(\mathcal{E})<\frac{2}{n}$, then
\begin{equation}\label{eq:pln}
\lim_{L\rightarrow\infty}p_{\mathsf{X}}(\mathcal{E}^{(L)})=\frac{1}{n}.
\end{equation}
\end{proposition}
\noindent For the case of $q_{\mathsf{X}}(\mathcal{E})<\frac{2}{n}$, we note that the right-hand side of Inequality~\eqref{eq:qgel} decreases to $\frac{1}{n}$ exponentially fast with respect to $L$.
Thus, the convergence \eqref{eq:pln} is also exponentially fast with respect to $L$.

From Proposition~\ref{prop:lqge} together with Inequalities~\eqref{eq:inq} and \eqref{eq:upb}, we have the following corollary.

\begin{corollary}\label{cor:pqn2}
For a two-state ensemble $\mathcal{E}=\{\eta_{i},\rho_{i}\}_{i=0}^{1}$ satisfying Condition \eqref{eq:trh} for a bipartition $\mathsf{X}\in\mathbb{X}_{2}$, we have
\begin{equation}\label{eq:pqn2}
p_{\mathsf{X}}(\mathcal{E}^{(L)})=q_{\mathsf{X}}(\mathcal{E}^{(L)})=\frac{1}{2}+\frac{1}{2}(2\eta_{0}-1)^{L}=\eta_{0}^{(L)}.
\end{equation}
\end{corollary}
\noindent Moreover, the equality in Eq.~\eqref{eq:pxpx} together with Inequality~\eqref{eq:inq} and Corollary~\ref{cor:pqn2} leads us to the following corollary.

\begin{corollary}\label{cor:apqn2}
For a two-state ensemble $\mathcal{E}=\{\eta_{i},\rho_{i}\}_{i=0}^{1}$ satisfying Condition \eqref{eq:trh} for all bipartitions, we have
\begin{equation}\label{eq:apqn2}
p_{\mathsf{X}}(\mathcal{E}^{(L)})=\eta_{0}^{(L)}
\end{equation}
for all nontrivial partitions $\mathsf{X}\in\mathbb{X}_{1}^{\mathsf{c}}$.
\end{corollary}

In information theory, data hiding is the protection of data by sharing it to multiple players so that the data only become accessible as a consequence of their cooperative action. 
By using particular two-party quantum states, there have been several quantum data-hiding schemes proposed to share the data between two players, Alice and Bob\cite{terh20011,divi2002,egge2002,mori2013,pian2014,peng2021,lami2021}.

Here, we establish a relation between a multi-party quantum state ensemble and its role in quantum data hiding. 
For a given multi-party quantum state ensemble $\mathcal{E}$, we first provide a condition in terms of $p_{\mathsf{G}}(\mathcal{E})$ and $q_{\mathsf{X}}(\mathcal{E})$ with respect to bipartitions $\mathsf{X}\in\mathbb{X}_{2}$. 
We show that $\mathcal{E}$ can be used to construct a multi-player quantum data-hiding scheme if $\mathcal{E}$ satisfies the condition.
The data are globally accessible, and non-global measurements can only provide arbitrary little information about the data. In other words, the data are perfectly hidden asymptotically.

For any integer $n\geqslant 2$, let us consider a multi-party quantum state ensemble $\mathcal{E}=\{\eta_{i},\rho_{i}\}_{i=0}^{n-1}$ satisfying the following condition
\begin{equation}\label{eq:pqgc}
p_{\mathsf{G}}(\mathcal{E})=1,~~
\max_{\mathsf{X}\in\mathbb{X}_{2}}q_{\mathsf{X}}(\mathcal{E})<\frac{2}{n}.
\end{equation}
By using the quantum state ensemble $\mathcal{E}$, we construct a quantum data-hiding scheme that can hide an $n$-ary classical data, that is, one in $\{0,\ldots,n-1\}$.

Hider, the extra party, first chooses a multi-party quantum state $\rho_{c_{1}}$ from $\mathcal{E}$ with the corresponding probability $\eta_{c_{1}}$ for $c_{1}\in\mathbb{Z}_{n}$ and shares it to the $m$ players $A_{1},\ldots,A_{m}$. 
The hider repeats this process $L$ times, so that $\rho_{c_{l}}$ is chosen with probability $\eta_{c_{l}}$ in the $l$th repetition for $l=1,\ldots,L$.
This whole procedure is equivalent to the situation that the hider first prepares the state 
\begin{equation}\label{eq:psm}
\rho_{\vec{c}}=\rho_{c_{1}}\otimes\cdots\otimes\rho_{c_{L}}
\end{equation}
from the $L$-fold ensemble $\mathcal{E}^{\otimes L}$ in Eq.~\eqref{eq:lfe}
with the corresponding probability $\eta_{\vec{c}}$ for $\vec{c}=(c_{1},\ldots,c_{L})\in\mathbb{Z}_{n}^{L}$, and sends it to the $m$ players $A_{1},\ldots,A_{m}$. 
This step is illustrated in Fig.~\ref{fig:qdh1}.

\begin{figure}[!tt]
\centerline{\includegraphics{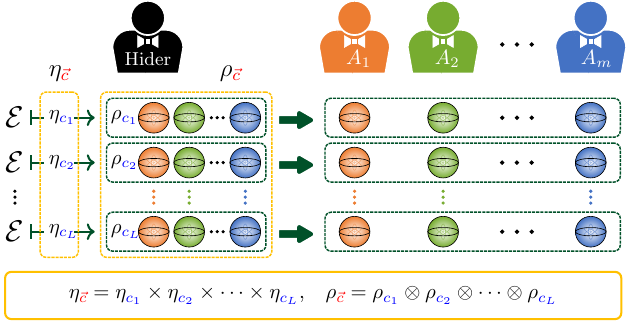}}
\caption{Data-hiding scheme based on a multi-party orthogonal quantum state ensemble  $\mathcal{E}=\{\eta_{i},\rho_{i}\}_{i=0}^{n-1}$ with $\max_{\mathsf{X}\in\mathbb{X}_{2}}q_{\mathsf{X}}(\mathcal{E})<\frac{2}{n}$ (Step~1). Hider first prepares the state $\rho_{\vec{c}}$ with the corresponding probability $\eta_{\vec{c}}$ from the $L$-fold ensemble of $\mathcal{E}$, and then sends $\rho_{\vec{c}}$ to the $m$ players $A_{1},\ldots,A_{m}$. 
}\label{fig:qdh1}
\end{figure}

To hide the classical data $x\in\mathbb{Z}_{n}$, the hider broadcasts the classical information $z$ to the $m$ players $A_{1},\ldots,A_{m}$,
\begin{equation}\label{eq:xyz}
z=x\oplus y
\end{equation}
where $\oplus$ is the modulo-$n$ addition and
\begin{equation}\label{eq:yonc}
y=\omega_{n}(\vec{c})
\end{equation}
for the prepared state $\rho_{\vec{c}}$ in Eq.~\eqref{eq:psm}.
We illustrate this step in Fig.~\ref{fig:qdh2}. 

\begin{figure}[!tt]
\centerline{\includegraphics{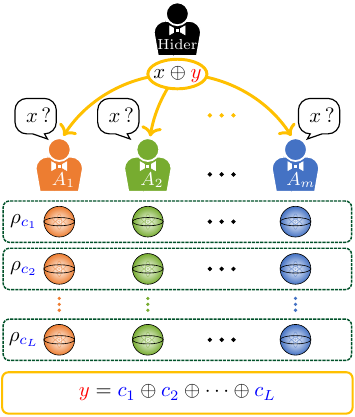}}
\caption{Data-hiding scheme based on a multi-party orthogonal quantum state ensemble $\mathcal{E}=\{\eta_{i},\rho_{i}\}_{i=0}^{n-1}$ with $\max_{\mathsf{X}\in\mathbb{X}_{2}}q_{\mathsf{X}}(\mathcal{E})<\frac{2}{n}$ (Step~2).
To hide the classical data $x\in\{0,\ldots,n-1\}$, hider broadcasts $x\oplus y$ to the $m$ players $A_{1},\ldots,A_{m}$, where $y$ is the modulo-$n$ summation of all entries in $\vec{c}$
for the prepared state $\rho_{\vec{c}}$.}\label{fig:qdh2}
\end{figure}

As the information $z$ in Eq.~\eqref{eq:xyz} was broadcasted to the $m$ players $A_{1},\ldots,A_{m}$, guessing the data $x$ is equivalent to guessing the data $y$. 
Moreover, guessing the data $y$ in Eq.~\eqref{eq:yonc} is equivalent to discriminating the states in Eq.~\eqref{eq:rhil} due to the argument in the paragraph containing Eqs.~\eqref{eq:mns} and \eqref{eq:rhil}.
For each partition $\mathsf{X}\in\mathbb{X}$, the maximum average probability of correctly guessing the data $x$ only using $\mathsf{X}$-LOCC measurements is equal to the maximum average probability of discriminating the states from the ensemble $\mathcal{E}^{(L)}$ in Eq.~\eqref{eq:eld} only using $\mathsf{X}$-LOCC measurements, that is, $p_{\mathsf{X}}(\mathcal{E}^{(L)})$.

The condition $p_{\mathsf{G}}(\mathcal{E})=1$ in Eq.~\eqref{eq:pqgc} means the states $\rho_{i}$'s of $\mathcal{E}$ are mutually orthogonal. 
In this case, the states $\rho_{\vec{c}}$'s of $\mathcal{E}^{\otimes L}$ in Eq.~\eqref{eq:erlf} are mutually orthogonal, which also implies the mutual orthogonality of $\rho_{i}^{(L)}$'s of $\mathcal{E}^{(L)}$ in Eq.~\eqref{eq:rhil}.
Thus, we have
\begin{equation}\label{eq:dpg1}
p_{\mathsf{G}}(\mathcal{E}^{(L)})=1.
\end{equation}
In other words, the data $x$ are accessible when the $m$ players $A_{1},\ldots,A_{m}$ can collaborate to use global measurements.

On the other hand, the condition $\max_{\mathsf{X}\in\mathbb{X}_{2}}q_{\mathsf{X}}(\mathcal{E})<\frac{2}{n}$ in Eq.~\eqref{eq:pqgc} together with Proposition~\ref{prop:lqge} and the equality in~\eqref{eq:pxpx} implies that $p_{\mathsf{X}}(\mathcal{E}^{(L)})$ can be arbitrarily close to $\frac{1}{n}$ by choosing an appropriately large $L$. 
In other words, Eq.~\eqref{eq:pln} holds for all nontrivial partitions $\mathsf{X}\in\mathbb{X}_{1}^{\mathsf{c}}$.

The following theorem states that the globally accessible data $x$ are perfectly hidden asymptotically if the $m$ players $A_{1},\ldots,A_{m}$ are only able to use non-global measurements.
\begin{theorem}\label{thm:ret}
Every multi-party quantum state ensemble $\mathcal{E}=\{\eta_{i},\rho_{i}\}_{i=0}^{n-1}$ with Condition \eqref{eq:pqgc} can be used to construct a multi-player data-hiding scheme that hides an $n$-ary classical data.
\end{theorem}

For a multi-party quantum state ensemble $\mathcal{E}=\{\eta_{i},\rho_{i}\}_{i=0}^{n-1}$, 
let us suppose that
$\eta_{0}<\frac{2}{n}$ and Condition \eqref{eq:trh} holds for all $\mathsf{X}\in\mathbb{X}_{2}$.
In this case, Theorem~\ref{thm:qxe0} implies that the condition $\max_{\mathsf{X}\in\mathbb{X}_{2}}q_{\mathsf{X}}(\mathcal{E})<\frac{2}{n}$ in Eq.~\eqref{eq:pqgc} holds; hence, we have the following corollary.

\begin{corollary}\label{cor:rer}
Every multi-party orthogonal quantum state ensemble $\mathcal{E}=\{\eta_{i},\rho_{i}\}_{i=0}^{n-1}$ satisfying $\eta_{0}<\frac{2}{n}$ and Condition \eqref{eq:trh} for all $\mathsf{X}\in\mathbb{X}_{2}$ can be used to construct a multi-player data-hiding scheme that hides an $n$-ary classical data. 
\end{corollary}

We also remark that our results can be used to hide the classical data by directly encoding it into a multi-party quantum state.
For any integers $n\geqslant2$, $L\geqslant1$ and a multi-party quantum state ensemble
$\mathcal{E}=\{\eta_{i},\rho_{i}\}_{i=0}^{n-1}$ satisfying Condition~\eqref{eq:pqgc}, let us consider the $n$ quantum states $\rho_{0}^{(L)},\ldots,\rho_{n-1}^{(L)}$ in Eq.~\eqref{eq:rhil}.
For a bipartition $\mathsf{X}\in\mathbb{X}_{2}$, the following proposition provides the convergence of $p_{\mathsf{X}}(\{\frac{1}{n},\rho_{i}^{(L)}\}_{i=0}^{n-1})$ as $L\rightarrow\infty$\cite{ha2024}.

\begin{proposition}\label{prop:povn}
For a bipartition $\mathsf{X}\in\mathbb{X}_{2}$ and a multi-party quantum state ensemble $\mathcal{E}=\{\eta_{i},\rho_{i}\}_{i=0}^{n-1}$ with $q_{\mathsf{X}}(\mathcal{E})<\frac{2}{n}$, we have
\begin{equation}\label{eq:plon}
\lim_{L\rightarrow\infty}p_{\mathsf{X}}(\{\tfrac{1}{n},\rho_{i}^{(L)}\}_{i=0}^{n-1})=\frac{1}{n}.
\end{equation}
\end{proposition}

To hide an $n$-ary classical data $x\in\mathbb{Z}_{n}$, hider sends the quantum state $\rho_{x}^{(L)}$ to the $m$ players $A_{1},\ldots,A_{m}$. 
This situation is illustrated in Fig.~\ref{fig:qdh3}, which is equivalent to discriminating the $n$ quantum states $\rho_{0}^{(L)},\ldots,\rho_{n-1}^{(L)}$ with the same probability $\frac{1}{n}$.
From the argument before Eq.~\eqref{eq:dpg1}, the condition $p_{\mathsf{G}}(\mathcal{E})=1$ in Eq.~\eqref{eq:pqgc} means that the states $\rho_{0}^{(L)},\ldots,\rho_{n-1}^{(L)}$ are mutually orthogonal.
Thus, we have 
\begin{equation}\label{eq:lrg}
p_{\mathsf{G}}(\{\tfrac{1}{n},\rho_{i}^{(L)}\}_{i=0}^{n-1})=1.
\end{equation}
In other words, the data $x$ are accessible when the $m$ players $A_{1},\ldots,A_{m}$ can collaborate to use global measurements.

\begin{figure}[!tt]
\centerline{\includegraphics{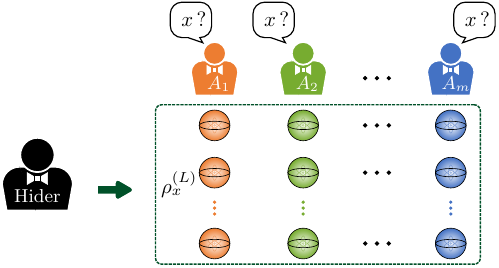}}
\caption{Data-hiding scheme using $n$ quantum states $\rho_{0}^{(L)},\ldots,\rho_{n-1}^{(L)}$ defined from a $m$-party orthogonal quantum state ensemble $\mathcal{E}=\{\eta_{i},\rho_{i}\}_{i=0}^{n-1}$ with $\max_{\mathsf{X}\in\mathbb{X}_{2}}q_{\mathsf{X}}(\mathcal{E})<\frac{2}{n}$.
To hide the $n$-ary classical data $x\in\{0,\ldots,n-1\}$, hider sends the quantum state $\rho_{x}^{(L)}$ to the $m$ players $A_{1},\ldots,A_{m}$.}\label{fig:qdh3}
\end{figure}

On the other hand, from Proposition~\ref{prop:povn} together with the equality in Eq.~\eqref{eq:pxpx} and Inequality~\eqref{eq:inq}, the condition $\max_{\mathsf{X}\in\mathbb{X}_{2}}q_{\mathsf{X}}(\mathcal{E})<\frac{2}{n}$ in Eq.~\eqref{eq:pqgc} implies  that 
Eq.~\eqref{eq:plon} holds for all nontrivial partitions $\mathsf{X}\in\mathbb{X}_{1}^{\mathsf{c}}$.
In other words, the globally accessible data $x$ are perfectly hidden asymptotically if the $m$ players $A_{1},\ldots,A_{m}$ can only perform non-global measurements. Thus, the $n$ quantum states $\rho_{0}^{(L)},\ldots,\rho_{n-1}^{(L)}$ can be used to construct a multi-player data-hiding scheme that hide the $n$-ary classical data by directly encoding them into the states.

\section{Examples}\label{sec:ex}
In this section, we provide examples of multi-party orthogonal quantum state ensembles to illustrate our results. 
First, we present an example of multi-party ensemble consisting of two quantum states to hide classical single-bit information.
\begin{example}\label{ex:npt}
For integers $d\geqslant2$ and $m\geqslant2$, let us consider the two orthogonal $m$-qu$d$it state ensemble $\mathcal{F}=\{\gamma_{i},\tau_{i}\}_{i=0}^{1}$ consisting of 
\begin{equation}\label{eq:ex1}
\gamma_{0}=\frac{d^{m}-1}{d^{m}},~
\tau_{0}=\frac{\mathbbm{1}-\ketbra{\mathrm{GHZ}}{\mathrm{GHZ}}}{d^{m}-1},~
\gamma_{1}=\frac{1}{d^{m}},~
\tau_{1}=\ketbra{\mathrm{GHZ}}{\mathrm{GHZ}}
\end{equation}
where
\begin{equation}\label{eq:ghzs}
\ket{\mathrm{GHZ}}=\frac{1}{\sqrt{d}}\sum_{i=0}^{d-1}\ket{i}_{A_{1}}\otimes\cdots\otimes\ket{i}_{A_{m}}
\end{equation}
is the $m$-qu$d$it Greenberger–Horne–Zeilinger state\cite{gree1989}.
\end{example}
For a bipartition $\mathsf{X}=\{\mathbf{A},\mathbf{B}\}\in\mathbb{X}_{2}$, it is straightforward to verify
\begin{eqnarray}\label{eq:exq1}
\gamma_{0}\Gamma_{\mathsf{X}}(\tau_{0})-\gamma_{1}\Gamma_{\mathsf{X}}(\tau_{1})
&=&\frac{1}{d^{m}}\Gamma_{\mathsf{X}}(\mathbbm{1}-2\ketbra{\mathrm{GHZ}}{\mathrm{GHZ}})\nonumber\\
&=&\frac{1}{d^{m}}\Big(\mathbbm{1}-\frac{2}{d}\sum_{i,j=0}^{d-1}\ket{i\cdots i}_{\mathbf{A}}\!\bra{j\cdots j}\otimes\ket{j\cdots j}_{\mathbf{B}}\!\bra{i\cdots i}\Big)
\end{eqnarray}
where
\begin{equation}\label{eq:iix}
\ket{i\cdots i}_{\mathbf{A}}=\bigotimes_{A\in \mathbf{A}}\ket{i}_{A},~
\ket{i\cdots i}_{\mathbf{B}}=\bigotimes_{B\in \mathbf{B}}\ket{i}_{B}.
\end{equation}

The summation term in Eq.~\eqref{eq:exq1} can be rewritten as
\begin{eqnarray}\label{eq:egrm}
\sum_{i,j=0}^{d-1}\ket{i\cdots i}_{\mathbf{A}}\!\bra{j\cdots j}\otimes\ket{j\cdots j}_{\mathbf{B}}\!\bra{i\cdots i}
=J+\sum_{i<j}\big(\ket{\Psi_{ij}^{+}}_{\mathsf{X}}\!\bra{\Psi_{ij}^{+}}-\ket{\Psi_{ij}^{-}}_{\mathsf{X}}\!\bra{\Psi_{ij}^{-}}\big)
\end{eqnarray}
where 
\begin{equation}\label{eq:igps}
J=\sum_{i=0}^{d-1}
\ket{i}_{A_{1}}\!\!\bra{i}
\otimes\cdots\otimes
\ket{i}_{A_{m}}\!\!\bra{i},~~
\ket{\Psi_{ij}^{\pm}}_{\mathsf{X}}
=\frac{1}{\sqrt{2}}\ket{i\cdots i}_{\mathbf{A}}\otimes\ket{j\cdots j}_{\mathbf{B}}
\pm\frac{1}{\sqrt{2}}\ket{j\cdots j}_{\mathbf{A}}\otimes\ket{i\cdots i}_{\mathbf{B}}.
\end{equation}
From Eqs.~\eqref{eq:exq1} and \eqref{eq:egrm}, we have
\begin{equation}\label{eq:srex1}
\gamma_{0}\Gamma_{\mathsf{X}}(\tau_{0})-\gamma_{1}\Gamma_{\mathsf{X}}(\tau_{1})
=\frac{1}{d^{m}}\Big(\mathbbm{1}-\frac{2}{d}J
-\frac{2}{d}\sum_{i<j}\ket{\Psi_{ij}^{+}}_{\mathsf{X}}\!\bra{\Psi_{ij}^{+}}
+\frac{2}{d}\sum_{i<j}\ket{\Psi_{ij}^{-}}_{\mathsf{X}}\!\bra{\Psi_{ij}^{-}}\Big)
\succeq0,
\end{equation}
where the positivity is due to $2/d\leqslant1$ and
\begin{equation}
\mathbbm{1}-J-\sum_{i<j}\ket{\Psi_{ij}^{+}}_{\mathsf{X}}\!\bra{\Psi_{ij}^{+}}\succeq0.
\end{equation}

Since the choice of $\mathsf{X}\in\mathbb{X}_{2}$ can be arbitrary, we have
\begin{equation}\label{eq:erp0}
\gamma_{0}\Gamma_{\mathsf{X}}(\tau_{0})-\gamma_{1}\Gamma_{\mathsf{X}}(\tau_{1})\succeq0
\end{equation}
for all bipartitions $\mathsf{X}\in\mathbb{X}_{2}$.
We further note that $\gamma_{0}<1=\frac{2}{n}$ for the case of two-state ensemble when $n=2$.
Thus, from Corollary~\ref{cor:rer}, the ensemble $\mathcal{F}$ in Eq.~\eqref{eq:ex1} can be used to construct a $m$-player data-hiding scheme that hides a classical bit. 

Now, we provide the following example of an ensemble consisting of $2^{t}$ states for an arbitrary positive integer $t$.
This ensemble of multi-party quantum states can be used to hide classical $t$-bit information.
\begin{example}\label{ex:gex}
For positive integers $d,m,n,s,t$ with $d,m\geqslant2$ and $n=2^{t}$,
let us consider the quantum state ensemble $\mathcal{E}=\{\eta_{i},\rho_{i}\}_{i=0}^{n-1}$ consisting of $t$-fold tensor products of $\sigma_{0}$ and $\sigma_{1}$,
\begin{equation}\label{eq:enx2}
\eta_{i}=\prod_{k=1}^{t}\lambda_{b_{k}(i)},~~
\rho_{i}=\bigotimes_{k=1}^{t}\sigma_{b_{k}(i)}
\end{equation}
where  
\begin{eqnarray}
\lambda_{0}=\frac{d^{ms}+(d^{m}-2)^{s}}{2d^{ms}},~
\sigma_{0}=\frac{\mathbbm{1}^{\otimes s}
+(\mathbbm{1}-2\ketbra{\mathrm{GHZ}}{\mathrm{GHZ}})^{\otimes s}}{d^{ms}+(d^{m}-2)^{s}},\nonumber\\
\lambda_{1}=\frac{d^{ms}-(d^{m}-2)^{s}}{2d^{ms}},~
\sigma_{1}=\frac{\mathbbm{1}^{\otimes s}
-(\mathbbm{1}-2\ketbra{\mathrm{GHZ}}{\mathrm{GHZ}})^{\otimes s}}{d^{ms}-(d^{m}-2)^{s}},\label{eq:sel}
\end{eqnarray}
and $b_{k}(i)$ is the $k$th digit in the $t$-digit binary representation of $i$,
\begin{equation}\label{eq:bki}
\sum_{k=1}^{t}b_{k}(i)\cdot 2^{k-1}=i.
\end{equation}
\end{example}

The probabilities $\lambda_{b}$'s and states $\sigma_{b}$'s in Eq.~\eqref{eq:sel} can be expressed using those in Eq.~\eqref{eq:ex1} as 
\begin{eqnarray}
\lambda_{b}=\sum_{\substack{\vec{c}\in\mathbb{Z}_{2}^{s}\\ \omega_{2}(\vec{c})=b}}
\gamma_{\vec{c}},~~
\sigma_{b}=\frac{1}{\lambda_{b}}
\sum_{\substack{\vec{c}\in\mathbb{Z}_{2}^{s}\\ \omega_{2}(\vec{c})=b}}
\gamma_{\vec{c}}\tau_{\vec{c}},~b\in\{0,1\}.\label{eq:eoet}
\end{eqnarray}
In other words, sharing the states $\sigma_{0}$ and $\sigma_{1}$ with the corresponding probabilities $\lambda_{0}$ and $\lambda_{1}$ in Eq.~\eqref{eq:sel} can be performed by $s$ times of sharing the states $\tau_{0}$ and $\tau_{1}$ with the corresponding probabilities $\gamma_{0}$ and $\gamma_{1}$ in Eq.~\eqref{eq:ex1}.
In this case, discriminating $\sigma_{0}$ and $\sigma_{1}$ is equivalent to discriminating the parity of $\omega_{2}(\vec{c})$ from the prepared state $\tau_{\vec{c}}$.
In this sense, we consider  $\rho_{i}$'s in Eq.~\eqref{eq:enx2} as $m$-party quantum states having  $d^{st}$-dimensional subsystems to be shared to the $m$ players $A_{1},\ldots,A_{m}$.
We also note that the states $\rho_{i}$'s in Eq.~\eqref{eq:enx2} are mutually orthogonal because $\sigma_{0}$ and $\sigma_{1}$ are orthogonal.

It is straightforward to verify
\begin{equation}\label{eq:erp}
\lambda_{0}\sigma_{0}=\frac{1}{2d^{ms}}(\Pi_{0}+\Pi_{1}),~~
\lambda_{1}\sigma_{1}=\frac{1}{2d^{ms}}(\Pi_{0}-\Pi_{1})
\end{equation}
where 
\begin{equation}\label{eq:opso}
\Pi_{0}=\mathbbm{1}^{\otimes s},~
\Pi_{1}=(\mathbbm{1}-2\ketbra{\mathrm{GHZ}}{\mathrm{GHZ}})^{\otimes s}.
\end{equation}
For each $i=0,\ldots,n-1$, it follows from Eqs.~\eqref{eq:enx2} and \eqref{eq:erp} that
\begin{equation}\label{eq:erb}
\eta_{i}\rho_{i}=\frac{1}{2d^{mst}}\sum_{\vec{a}\in\mathbb{Z}_{2}^{t}}
\bigotimes_{k=1}^{t}(-1)^{a_{k} b_{k}(i)}\Pi_{a_{k}}
=\frac{1}{2d^{mst}}\sum_{\vec{a}\in\mathbb{Z}_{2}^{t}}
(-1)^{h_{\vec{a}}(i)}\Pi_{\vec{a}}
\end{equation}
where 
\begin{equation}\label{eq:piva}
h_{\vec{a}}(i)=\sum_{k=1}^{t}a_{k} b_{k}(i),~~
\Pi_{\vec{a}}=\bigotimes_{k=1}^{t}\Pi_{a_{k}},~\vec{a}=(a_{1},\ldots,a_{t}).
\end{equation}

To use Corollary~\ref{cor:rer}, let $\mathsf{X}$ be an arbitrary bipartition in $\mathbb{X}_{2}$.
We can see that
\begin{equation}\label{eq:pop}
\Gamma_{\mathsf{X}}(\Pi_{0})=\mathbbm{1}^{\otimes s}\succeq0,~~
\Gamma_{\mathsf{X}}(\Pi_{1})=\Gamma_{\mathsf{X}}(\mathbbm{1}-2\ketbra{\mathrm{GHZ}}{\mathrm{GHZ}})^{\otimes s}\succeq0,
\end{equation}
where the second positivity is from Eqs.~\eqref{eq:exq1} and \eqref{eq:srex1}.
For each $i=0,\ldots,n-1$, we have
\begin{equation}\label{eq:emnp}
\eta_{0}\Gamma_{\mathsf{X}}(\rho_{0})-\eta_{i}\Gamma_{\mathsf{X}}(\rho_{i})
=\frac{1}{2d^{mst}}\sum_{\vec{a}\in\mathbb{Z}_{2}^{t}}
\Big[1-(-1)^{h_{\vec{a}}(i)}\Big]\Gamma_{\mathsf{X}}(\Pi_{\vec{a}})\succeq0,
\end{equation}
where the equality is due to Eq.~\eqref{eq:erb} and the positivity follows from the positivities in \eqref{eq:pop}.
Since the choice of $\mathsf{X}\in\mathbb{X}_{2}$ can be arbitrary, Condition~\eqref{eq:trh} holds for all bipartitions.

We also note that $\eta_{0}<\frac{2}{n}$ for
\begin{equation}\label{eq:ifoi}
\Big(1-\frac{2}{d^{m}}\Big)^{s}<\sqrt[t]{2}-1
\end{equation}
because
\begin{equation}\label{eq:qgq}
\eta_{0}=\lambda_{0}^{t}=\frac{1}{2^{t}}\Big[1+\Big(1-\frac{2}{d^{m}}\Big)^{s}\Big]^{t},~~
\frac{2}{n}=\frac{2}{2^{t}}.
\end{equation}
Thus, from Corollary~\ref{cor:rer}, the quantum state ensemble $\mathcal{E}$ in Eq.~\eqref{eq:enx2} satisfying Inequality~\eqref{eq:ifoi} can be used to construct a $m$-player data-hiding scheme that hides $t$ classical bits.

Even if the desired state ensemble cannot be perfectly realized in practice and the prepared ensemble deviates slightly, our data-hiding scheme remains effective as long as the resulting ensemble satisfies the inequality condition in Eq.~\eqref{eq:pqgc}. 
The primary issue arises when even slight noise or errors disturb the equality condition in Eq.~\eqref{eq:pqgc}, thereby undermining the guarantee of perfect data recovery through global measurements. 
In these cases, increasing the number of state preparation repetitions can further weaken the reliability of global recovery. Fortunately, since our scheme requires only a small number of repetitions, the data recoverability through global measurements is hardly affected.

\section{Conclusions}\label{sec:dsc}
We have provided multi-player quantum data hiding based on nonlocal quantum state ensembles arising from multi-party quantum state discrimination.
Using the bounds on local minimum-error discrimination of multi-party quantum state ensemble, we have established a sufficient condition for a multi-party orthogonal quantum state ensemble to be used to construct a multi-player data-hiding scheme; to make a classical data accessible only by global measurements, the data are perfectly hidden asymptotically(Theorem~\ref{thm:ret} and Corollary~\ref{cor:rer}).
Our results have been illustrated by examples in multi-party quantum systems(Examples~\ref{ex:npt} and \ref{ex:gex}).

Even if the ideal state ensemble cannot be perfectly achieved in practice and minor deviations occur, our data-hiding protocol remains robust as long as the ensemble satisfies the inequality condition given in Eq.~\eqref{eq:pqgc}. Furthermore, because our method requires only a few repetitions to sufficiently limit the information accessible through LOCC measurements, it provides a highly practical approach for the experimental realization of quantum data hiding.

Our results tells us that any multi-party quantum state ensemble satisfying Condition~\eqref{eq:pqgc} is useful in multi-player quantum data hiding.
Therefore, it is natural to ask whether Condition~\eqref{eq:pqgc} is also necessary for an ensemble to be used in quantum data hiding. 
This question is related to the research on the operational meaning of the quantity defined in Eq.~\eqref{eq:qge} in terms of quantum data hiding.
It is also an interesting future work to generalize our results to $(m,k)$-threshold data-hiding schemes where the hidden classical data can be recovered if $k$ or more players collaborate.

\section*{Acknowledgments}
This work was supported by a National Research Foundation of Korea(NRF) grant funded by the Korean government(Ministry of Science and ICT)(No.NRF2023R1A2C1007039). JSK was supported by Creation of the Quantum Information Science R\&D Ecosystem (Grant No. 2022M3H3A106307411) through the National Research Foundation of Korea (NRF) funded by the Korean government (Ministry of Science and ICT).



\begin{thebibliography}{99}
\bibitem{pptp}  
The eigenvalues of a partially transposed operator do not depend on the choice of basis or the subsystem to be transposed. 
For simplicity, we consider the standard basis and the second system throughout this paper.
\end{thebibliography}
\end{document}